\documentclass{article}
\usepackage{amsfonts}
\usepackage{makeidx}
\usepackage{amsmath}
\usepackage{amssymb}
\usepackage{graphicx}

\newtheorem{theorem}{Theorem}

\newtheorem{corollary}[theorem]{Corollary}

\newtheorem{lemma}[theorem]{Lemma}

\newtheorem{proposition}[theorem]{Proposition}

\newenvironment{proof}[1][Proof]{\noindent\textbf{#1.} }{\ \rule{0.5em}{0.5em}}
\input{tcilatex}
\begin{document}

\title{{\Large \textbf{Lie point symmetries of a general class of PDEs: The
heat equation}}}
\author{Andronikos Paliathanasis\thanks{%
Email: anpaliat@phys.uoa.gr} \ and Michael Tsamparlis\thanks{%
Email: mtsampa@phys.uoa.gr} \\
%EndAName
{\small \textit{Faculty of Physics, Department of Astrophysics - Astronomy -
Mechanics,}}\\
{\small \textit{\ University of Athens, Panepistemiopolis, Athens 157 83,
GREECE}}}
\date{}
\maketitle

\begin{abstract}
We give two theorems which show that the Lie point and the Noether
symmetries of a second-order ordinary differential equation of the form $%
\frac{D}{Ds}\left( \frac{Dx^{i}\left( s\right) }{Ds}\right) =F(x^{i}\left(
s\right) ,\dot{x}^{j}\left( s\right) )$ are subalgebras of the special
projective and the homothetic algebra of the space respectively. We examine
the possible extension of this result to partial differential equations
(PDE) of the form $A^{ij}u_{ij}-F(x^{i},u,u_{i})=0$ where $u(x^{i})$ and $%
u_{ij}$ stands for the second partial derivative. We find that if the
coefficients $A^{ij}$ are independent of $u(x^{i})$ then the Lie point
symmetries of the PDE form a subgroup of the conformal symmetries of the
metric defined by the coefficients $A^{ij}$. We specialize the study to
linear forms of $F(x^{i},u,u_{i})$ and write the Lie symmetry conditions for
this case. We apply this result to two cases. The wave equation in an
inhomogeneous medium for which we derive the Lie symmetry vectors and check
our results with those in the literature. Subsequently we consider the heat
equation with a flux in an $n-$dimensional Riemannian space and show that
the Lie symmetry algebra is a subalgebra of the homothetic algebra of the
space. We discuss this result in the case of de Sitter space time and in
flat space.
\end{abstract}


PACS - numbers: 2.40.Hw, 4.20.-q, 4.20.Jb, 04.20.Me, 03.20.+i, 02.40.Ky

\section{Introduction}

\label{Introduction}

A Lie point symmetry of an ordinary differential equation (ODE) is a point
transformation in the space of variables which preserves the set of
solutions of the ODE \cite{Bluman book ODES,Olver Book,Stephani book ODES}.
If we look at these solutions as curves in the space of variables, then we
may equivalently consider a Lie point symmetry as a point transformation
which preserves the set of solution curves. Applying this observation to the
geodesic curves in \ a Riemannian (affine) space, we infer that the Lie
point symmetries of the geodesic equations in any Riemannian (affine) space
are the automorphisms which preserve the set of these curves. However, it is
known by Differential Geometry that the point transformations of a
Riemannian (affine)\ space which preserve the set of geodesics are the
projective transformations. Therefore it is reasonable to expect a
correspondence between the Lie symmetries of the geodesic equations and the
projective algebra of the space.

The equation of geodesics in an arbitrary coordinate frame is a second-order
ODE of the form%
\begin{equation}
\ddot{x}^{i}+\Gamma _{jk}^{i}\dot{x}^{j}\dot{x}^{k}+F(x^{i},\dot{x}^{j})=0
\label{de.0}
\end{equation}%
where $F(x^{i},\dot{x}^{j})$ is an arbitrary function of its arguments and
the functions $\Gamma _{jk}^{i}$ are the connection coefficients of the
space.\footnote{%
This point of view can be generalized to the general second order ODE
provided the functions $\Gamma _{jk}^{i}$ can be identified with the
connection coefficients of a metric.} Equivalently equation (\ref{de.0}) is
the equation of motion of a dynamical system moving in a Riemannian (affine)
space under the action of a velocity dependent force. According to the above
argument we expect that the Lie symmetries of the ODE (\ref{de.0}) for a
given function $F(x^{i},\dot{x}^{j})$ are a subalgebra of the projective
algebra of the space. This subalgebra is selected by means of certain
constraint conditions which will involve geometric quantities of the space
and the function $F(x^{i},\dot{x}^{j})$. This approach is not new and
similar considerations can be found in \cite{Prince Crampin (1984) 1,Aminova
2006,Aminova2010,FerozeMahomedQadir,TsamparlisGRG,TsamparlisGRG2}.

The determination of the Lie point symmetries of a given system of ODEs
consists of two steps (a)\ the determination of the conditions which the
components of the Lie symmetry vectors must satisfy and (b) the solution of
the system of these conditions. Step (a) is formal and it is outlined in
e.g. \cite{Bluman book ODES,Olver Book,Stephani book ODES}. The second step
is the key one and, for example, in higher dimensions where one has a large
number of simultaneous equations the solution can be quite involved.
However, if we express the system of Lie symmetry conditions of (\ref{de.0})
in a Riemannian space in terms of collineation (i.e. symmetry) conditions of
the metric, then the determination of Lie symmetries is transferred to the
geometric problem of determining the generators of the projective group of
the metric. In this field there is a vast amount of work that is already
done to be used. Indeed the projective symmetries are already known in many
cases or they can be determined by existing general theorems of Differential
Geometry. For example the projective algebra and all its subalgebras are
known for the spaces of constant curvature \cite{Barnes} and in particular
for the flat spaces. This implies, for example, that the Lie symmetries of 
\emph{all} Newtonian dynamical systems are "known" and the same applies to
dynamical systems in Special Relativity!

In this work we state a theorem which establishes the exact relation between
the projective algebra of the space and the Lie symmetry algebra of (\ref%
{de.0}), assuming that the function $F$ depends only on the coordinates,
i.e. $F(x^{i}).$

What has been said for the Lie point symmetries of (\ref{de.0}) applies also
to Noether symmetries (provided (\ref{de.0}) follows from a Lagrangian). The
Noether symmetries are Lie point symmetries which satisfy the constraint%
\begin{equation}
X^{\left[ 1\right] }L+L\frac{d\xi }{dt}=\frac{df}{dt}.  \label{L2p.3}
\end{equation}%
Noether symmetries form a closed subalgebra of the Lie symmetries algebra.
In accordance to the above, this implies that the Noether symmetries will be
related with a subalgebra of the projection algebra of the space where
`motion' occurs. As it will be shown, this subalgebra is contained in the
homothetic algebra of the space.

As it is well known, each Noether point symmetry is associated conserved
current (i.e. first integral), hence the above imply that the (standard)
conserved quantities of a dynamical system depend on the space it moves and
the type of force $F(x^{i},\dot{x}^{j})$ which modulates the motion. In
particular, in `free fall', i.e. in the case $F(x^{i},\dot{x}^{j})=0$ the 
\emph{geometry} of the space is the sole factor which determines the
(standard) first integrals of motion. This conclusion is by no means trivial
and shows the deep relation between Geometry and Physics!

A natural question which arises is the following:

\emph{To what extend this correspondence of Lie/Noether symmetries of second
order ODES of the form (\ref{de.0}), with the collineations of the space, is
extendable to partial differential equations of second-order of a similar
form?}

Obviously, a global answer to this question is not possible. However, it can
be shown that for many interesting PDEs the Lie symmetries are indeed
obtained from the collineations of the metric. Pioneering work in this
direction is the work of Ibragimov \cite{Ibragimov book 1}. Recently,
Bozhkov et al. \cite{Boskov} studied the Lie and the Noether symmetries of
the Poisson equation and shown that the Lie symmetries of the Poisson PDE
are generated from the conformal algebra of the metric.

In the present work we show that for a general class of PDEs of
second-order, there is a close relation between the Lie symmetries and the
conformal algebra of the space. Subsequently we apply these results to a
number of interesting PDEs and regain existing results in a unified manner.
As a new application, we determine the Lie symmetries of the heat equation
with a flux in a $n-$dimensional Riemannian space.

The structure of the paper is as follows. In Section \ref{Collineations of
Riemannian spaces} we review briefly the collineations in a Riemannian
space. In Section \ref{The symmetry conditions using Lie symmetry methods}
we consider the equation of geodesics in an affine space and determine the
Lie symmetry conditions in covariant form. We find that the major symmetry
condition relates the Lie symmetries with the special projective algebra of
the space. A similar result has been obtained previously in \cite{Prince
Crampin (1984) 1} using the bundle formulation of second order ODEs.

In Section \ref{The conservative system} we solve, in a concise manner, the
symmetry conditions and state Theorem \ref{The general conservative system}
which gives the Lie symmetry vectors in terms of the collineations of the
metric and Theorem \ref{The Noether Theorem} which gives the Noether point
symmetries in terms of the homothetic algebra of the metric.

In Section \ref{The case of the second order PDE's} we consider the PDEs of
the form $A^{ij}u_{ij}-F(x^{i},u,u_{i})=0$ and derive the Lie symmetry
conditions. We show that for these PDEs the $\xi ^{i}(x^{j})$ provided $%
A^{ij}\neq 0$ and if $A_{,u}^{ij}=0$ then $\xi ^{i}(x^{j})\partial _{i}$ is
a Conformal Killing vector of the metric $A^{ij}$.

In Section \ref{The Lie symmetry conditions for a linear function} we
consider a linear form of $F(x^{i},u,u_{i})$ and determine the Lie symmetry
conditions in geometric form. In section \ref{Applications} we apply the
results of section \ref{The Lie symmetry conditions for a linear function}
to the wave equation in a two dimensional inhomogeneous medium and to the
heat equation with flux in a $n-$dimensional Riemannian space. In the latter
case, it is shown that the Lie symmetry vectors are obtained form the
homothetic algebra of the $A^{ij}$ metric. A similar result has been found
for the Poisson equation \cite{Boskov}.

Finally in Section \ref{Conclusions} we comment on the results and point out
possible directions for future research.

\section{Collineations of Riemannian spaces}

\label{Collineations of Riemannian spaces}

A\ collineation in a Riemannian space is a vector field $\mathbf{X}$ which
satisfies an equation of the form%
\begin{equation}
\mathcal{L}_{X}\mathbf{A=B}  \label{L2p.2}
\end{equation}%
where $\mathcal{L}_{X}$ denotes Lie derivative, $\mathbf{A}$ is a geometric
object (not necessarily a tensor)\ defined in terms of the metric and its
derivatives (e.g. connection, Ricci tensor, curvature tensor etc.) and $%
\mathbf{B}$ is an arbitrary tensor with the same tensor indices as $\mathbf{A%
}$. The collineations in a Riemannian space have been classified by Katzin
et al. \cite{Katzin}. In the following we use only certain collineations.

A conformal Killing vector (CKV)\ is defined by the relation 
\begin{equation}
\mathcal{L}_{X}g_{ij}=2\psi \left( x^{k}\right) g_{ij}.
\end{equation}%
If $\psi =0,$ $\mathbf{X}\ $is called a Killing vector (KV). If $\psi $ is a
nonvanishing constant then $\mathbf{X~}$is a homothetic vector (HV) and if $%
\psi _{;ij}=0,$ $\mathbf{X}$ is a special conformal Killing vector (SCKV). A
CKV is called proper if it is not a KV, a HV or a SCKV.

A Projective collineation (PC) is defined by the equation 
\begin{equation}
\mathcal{L}_{X}\Gamma _{jk}^{i}=2\phi _{(,j}\delta _{k)}^{i}.
\end{equation}%
If $\phi =0,$ the PC\ is called an affine collineation (AC) and if $\phi
_{;ij}=0,$ a special projective collineation (SPC). A\ proper PC\ is a PC\
which is not an AC, HV or KV or SPC. The PCs form a Lie algebra whose ACs,
HV and KVs form subalgebras. It has been shown that if a metric admits a
SCKV, then also admits a SPC, a gradient HV and a gradient KV \cite{HallR}.

In the following we shall need the symmetry algebra of spaces of constant
curvature. In \cite{Barnes} it has been shown that the PCs of a space of
constant nonvanishing curvature consist of proper PCs and KVs only and if
the space is flat then the algebra of the PCs consists of KVs/HV/ACs and
SPCs.

In particular, for the Euclidian space $E^{n}$ the projection algebra
consists of the vectors of in Table 1.

\begin{center}
.%
\begin{tabular}{|l|l|l|}
\multicolumn{3}{l}{Table 1: Collineations of Euclidean space $%
E^{n}~,~I,J=1,2,\ldots ,n$.} \\ \hline
Collineation & Gradient & Nongradient \\ \hline
\multicolumn{1}{|l|}{Killing vectors (KV)} & $\mathbf{S}_{I}=\delta
_{I}^{i}\partial _{i}$ & $\mathbf{X}_{IJ}=\delta _{\lbrack I}^{j}\delta
_{J]}^{i}x_{j}\partial _{i}$ \\ \hline
\multicolumn{1}{|l|}{Homothetic vector (HV)} & $\mathbf{H}=x^{i}\partial
_{i}~$ &  \\ \hline
\multicolumn{1}{|l|}{Affine Collineation (AC)} & $\mathbf{A}%
_{II}=x_{I}\delta _{I}^{i}\partial _{i}~$ & $\mathbf{A}_{IJ}=x_{J}\delta
_{I}^{i}\partial _{i}$ \\ \hline
\multicolumn{1}{|l|}{Special Projective collineation (SPC)} &  & $\mathbf{P}%
_{I}=S_{I}\mathbf{H}.~$ \\ \hline
\end{tabular}
\end{center}

\section{The Lie point symmetry conditions in an affine space}

\label{The symmetry conditions using Lie symmetry methods}

We consider the system of ODEs: 
\begin{equation}
\ddot{x}^{i}+\Gamma _{jk}^{i}\dot{x}^{j}\dot{x}^{k}+\sum%
\limits_{m=0}^{n}P_{j_{1}...j_{m}}^{i}\dot{x}^{j_{1}}\ldots \dot{x}^{j_{m}}=0
\label{de.1}
\end{equation}%
where $\Gamma _{jk}^{i}$ are the connection coefficients of the space and $%
P_{j_{1}...j_{m}}^{i}(t,x^{i})$ are smooth polynomials completely symmetric
in the lower indices and derive the Lie point symmetry conditions in
geometric form using the standard approach. Equation (\ref{de.1}) is quite
general and covers most of the standard cases, autonomous and non
autonomous, and in particular equation (\ref{de.0}). Furthermore because the 
$\Gamma _{jk}^{i}$'s are not assumed to be symmetric, the results are valid
in a space with torsion. Obviously they hold in a Riemannian space provided
that the connection coefficients are given in terms of the Christofell
symbols.

The detailed calculation has been given in \cite{TsamparlisGRG2} and shall
not be repeated here. In the following we summarize these results.

The terms $\dot{x}^{j_{1}}\ldots \dot{x}^{j_{m}}$ for $m\leq 4$ give the
equations: 
\begin{align}
L_{\eta }P^{i}+2\xi ,_{t}P^{i}+\xi P^{i},_{t}+\eta ^{i},_{tt}+\eta
^{j},_{t}P_{.j}^{i}& =0 \\
L_{\eta }P_{j}^{i}+\xi ,_{t}P_{j}^{i}+\xi P_{j}^{i},_{t}+\left( \xi
,_{k}\delta _{j}^{i}+2\xi ,_{j}\delta _{k}^{i}\right) P^{k}+2\eta
^{i},_{t|j}-\xi ,_{tt}\delta _{k}^{i}+2\eta ^{k},_{t}P_{.jk}^{i}& =0 \\
L_{\eta }P_{jk}^{i}+L_{\eta }\Gamma _{jk}^{i}+\left( \xi ,_{d}\delta
_{(k}^{i}+\xi ,_{(k}\delta _{|d|}^{i}\right) P_{.j)}^{d}+\xi
P_{.kj,t}^{i}-2\xi ,_{t(j}\delta _{k)}^{i}+3\eta ^{d},_{t}P_{.dkj}^{i}& =0 \\
L_{\eta }P_{.jkd}^{i}-\xi ,_{t}P_{.jkd}^{i}+\xi ,_{e}\delta
_{(k}^{i}P_{.dj)}^{e}+\xi P_{.jkd,t}^{i}+4\eta ^{e},_{t}P_{.jkde}^{i}-\xi
_{(,j|k}\delta _{d)}^{i}& =0~
\end{align}%
and the conditions due to the terms $\dot{x}^{j_{1}}\ldots \dot{x}^{j_{m}}$
for $m>4$ are given by the following general formula: 
\begin{align}
& L_{\eta }P_{j_{1}...j_{m}}^{i}+P_{j_{1}...j_{m}~,t}^{i}\xi +\left(
2-m\right) \xi _{,t}P_{j_{1}...j_{m}}^{i}+  \notag \\
& +\xi _{,r}\left( 2-\left( m-1\right) \right) P_{j_{1}...j_{m-1}}^{i}\delta
_{j_{m}}^{r}+\left( m+1\right) P_{j_{1}...j_{m+1}}^{i}\eta
_{,t}^{j_{m+1}}+\xi _{,j}P_{j_{1}...j_{m-1}}^{j}\delta _{j_{m}}^{i}=0.
\end{align}

We note the appearance of the term $L_{\eta }\Gamma _{jk}^{i}$ in these
expressions.

Eqn (\ref{de.0}) is obtained for $m=0,$ $~P^{i}=F^{i}$ \ in which case the
Lie symmetry conditions read: 
\begin{align}
L_{\eta }P^{i}+2\xi ,_{t}P^{i}+\xi P^{i},_{t}+\eta ^{i},_{tt}& =0
\label{de.13} \\
\left( \xi ,_{k}\delta _{j}^{i}+2\xi ,_{j}\delta _{k}^{i}\right) P^{k}+2\eta
^{i},_{t|j}-\xi ,_{tt}\delta _{k}^{i}& =0  \label{de.14} \\
L_{\eta }\Gamma _{jk}^{i}-2\xi ,_{t(j}\delta _{k)}^{i}& =0  \label{de.15} \\
\xi _{(,j|k}\delta _{d)}^{i}& =0.  \label{de.16}
\end{align}%
If $F^{i}=0$ we obtain the Lie symmetry conditions for the geodesic
equations (see \cite{TsamparlisGRG2}) .

\section{The autonomous dynamical system moving in a Riemannian space}

\label{The conservative system}

We `solve' the Lie symmetry conditions (\ref{de.13}) - (\ref{de.16}) for an
autonomous dynamical system in the sense that we express them in terms of
the collineations of the metric.

Equation (\ref{de.16}) means that $\xi _{,j}$ is a gradient KV of $g_{ij}.$
This implies that the metric $g_{ij}$ is decomposable. Equation (\ref{de.15}%
) means that $\eta ^{i}$ is a projective collineation of the metric with
projective function $\xi _{,t}.$ The remaining two equations are the
constraint conditions, which relate the components $\xi ,n^{i}$ of the Lie
symmetry vector with the vector $F^{i}(x^{j})$. Equation (\ref{de.13}) gives 
\begin{equation}
\left( L_{\eta }g^{ij}\right) F_{j}+g^{ij}L_{\eta }F_{j}+2\xi
_{,t}g^{ij}F_{j}+\eta _{,tt}^{i}=0.  \label{de.21a}
\end{equation}%
This equation is an additional restriction for $\eta ^{i}$ because it
relates it directly to the metric symmetries. Finally equation (\ref{de.14})
gives%
\begin{equation}
-\delta _{j}^{i}\xi _{,tt}+\left( \xi _{,j}\delta _{k}^{i}+2\delta
_{j}^{i}\xi _{,k}\right) F^{k}+2\eta _{,tj}^{i}+2\Gamma _{jk}^{i}\eta
_{,t}^{k}=0.  \label{de.21d}
\end{equation}

We conclude that the Lie symmetry equations are equations (\ref{de.21a}) ,(%
\ref{de.21d}) where $\xi (t,x)$ is a gradient KV of the metric $g_{ij}$ and $%
\eta ^{i}\left( t,x\right) $ is a special Projective collineation of the
metric $g_{ij}$ with projective function $\xi _{,t}$. We state the results
in theorem \ref{The general conservative system} \cite{TsamparlisGRG2}.

\begin{theorem}
\label{The general conservative system} The Lie point symmetries of the
system of equations of motion of an autonomous system under the action of
the force $F^{j}(x^{i})$ in a general Riemannian space with metric $g_{ij},$
namely%
\begin{equation}
\ddot{x}^{i}+\Gamma _{jk}^{i}\dot{x}^{j}\dot{x}^{k}=F^{i}  \label{PP.01}
\end{equation}%
are given in terms of the generators $Y^{i}$ of the special projective
algebra of the metric $g_{ij}.$
\end{theorem}

\bigskip If the force $F^{i}$ is derivable from a potential $V(x^{i})$ so
that the equations of motion follow from the standard Lagrangian 
\begin{equation}
L\left( x^{j},\dot{x}^{j}\right) =\frac{1}{2}g_{ij}\dot{x}^{i}\dot{x}%
^{j}-V\left( x^{j}\right)  \label{NPC.02}
\end{equation}%
with Hamiltonian%
\begin{equation}
E=\frac{1}{2}g_{ij}\dot{x}^{i}\dot{x}^{j}+V\left( x^{j}\right) ~
\label{NPC.3}
\end{equation}%
the Noether conditions, are 
\begin{eqnarray}
V_{,k}\eta ^{k}+V\xi _{,t} &=&-f_{,t} \\
\eta _{,t}^{i}g_{ij}-\xi _{,j}V &=&f_{,j} \\
L_{\eta }g_{ij} &=&2\left( \frac{1}{2}\xi _{,t}\right) g_{ij} \\
\xi _{,k} &=&0.
\end{eqnarray}

Last equation implies $\xi =\xi \left( t\right) $ and reduces the system as
follows%
\begin{eqnarray}
L_{\eta }g_{ij} &=&2\left( \frac{1}{2}\xi _{,t}\right) g_{ij}  \label{NPC.4}
\\
V_{,k}\eta ^{k}+V\xi _{,t} &=&-f_{,t}  \label{PP.01.6} \\
\eta _{i,t} &=&f_{,i}.  \label{PP.01.7}
\end{eqnarray}

Equation (\ref{NPC.4}) implies that $\eta ^{i}$ is a conformal Killing
vector of the metric provided $\xi _{,t}\neq 0.$ Because $g_{ij}$\ is
independent of $t$\ and $\xi =\xi \left( t\right) $\ the $\eta ^{i}$\ must
be is a HV of the metric. This means that $\eta ^{i}\left( t,x\right)
=T\left( t\right) Y^{i}\left( x^{j}\right) $\ where $Y^{i}$\ is a HV. If $%
\xi _{,t}=0$ then $\eta ^{i}$ is a Killing vector of the metric. Equations (%
\ref{PP.01.6}), (\ref{PP.01.7}) are the constraint conditions, which the
Noether symmetry and the potential must satisfy for the former to be
admitted. These lead to the following theorem \cite{TsamparlisGRG2}.

\begin{theorem}
\label{The Noether Theorem}The Noether point symmetries of the Lagrangian (%
\ref{NPC.02}) are generated from the homothetic algebra of the metric $%
g_{ij} $.
\end{theorem}

More specifically, concerning the Noether symmetries, we have the following

All autonomous systems admit the Noether symmetry $\partial _{t}~$whose
Noether integral is the Hamiltonian~$E$ (\ref{NPC.3}). For the rest of the
Noether symmetries we consider the following cases

\textbf{Case I }\ Noether point symmetries generated by the homothetic
algebra.

The Noether symmetry vector and the Noether function $G\left( t,x^{k}\right) 
$ are%
\begin{equation}
\mathbf{X}=2\psi _{Y}t\partial _{t}+Y^{i}\partial _{i}~,~G\left(
t,x^{k}\right) =pt  \label{NPC.03}
\end{equation}%
where, $\psi _{Y}$ is the homothetic factor of $Y^{i}~$($\psi _{Y}=0$ for a
KV\ and $1$ for the HV) and $p$ is a constant, provided the potential
satisfies the condition%
\begin{equation}
\mathcal{L}_{Y}V+2\psi _{Y}V+p=0.  \label{NPC.04}
\end{equation}

\textbf{Case II} \ Noether point symmetries generated by the gradient
homothetic Lie algebra, i.e., both KVs and the HV are gradient. \ 

In this case the Noether symmetry vector and the Noether function are%
\begin{equation}
\mathbf{X}=2\psi _{Y}\int T\left( t\right) dt\partial _{t}+T\left( t\right)
H^{i}\partial _{i}~~,~G\left( t,x^{k}\right) =T_{,t}H\left( x^{k}\right)
~+p\int Tdt  \label{NPC.05}
\end{equation}%
where, $H^{i}$ is the gradient HV or a gradient KV, the function $T(t)$ is
computed from the relation~$~T_{,tt}=mT~\ $where $~m$ is a constant and the
potential satisfies the condition 
\begin{equation}
\mathcal{L}_{H}V+2\psi _{Y}V+mH+p=0.  \label{NPC.06}
\end{equation}

Concerning the Noether integrals we have the following result (not including
the Hamiltonian)

\begin{corollary}
\label{The Noether Integrals}The Noether integrals of Case I and Case II are
respectively 
\begin{equation}
I_{C_{I}}=2\psi _{Y}tE-g_{ij}Y^{i}\dot{x}^{j}+pt  \label{NPC.07}
\end{equation}%
\begin{equation}
I_{C_{II}}=2\psi _{Y}\int T\left( t\right) dt~E-g_{ij}H^{,i}\dot{x}%
^{j}+T_{,t}H+p\int Tdt  \label{NPC.08}
\end{equation}%
where, $E$ is the Hamiltonian (\ref{NPC.3}).
\end{corollary}

We remark that theorems \ref{The general conservative system} and \ref{The
Noether Theorem} do not apply to generalized symmetries\cite{Sarlet,Kalotas}.

In a number of recent papers \ \cite%
{FerozeMahomedQadir,Feroze2011,Hussain2010,Feroze2010}, the authors study
the relation between the Noether symmetries of the geodesic Lagrangian 
\begin{equation}
L=\frac{1}{2}g_{ij}\dot{x}^{i}\dot{x}^{j}  \label{L1}
\end{equation}%
where $\dot{x}^{a}=\frac{dx^{a}}{ds}$ ($s$ is an affine parameter along the
geodesics) with the spacetime symmetries. They also make a conjecture
concerning the relation between the Noether symmetries and the conformal
algebra of spacetime and concentrate especially on conformally flat
spacetimes. In \cite{Feroze2010} it is also claimed that the author has
found new conserved quantities for spaces of different curvatures, which
seem to be of nonnoetherian character. Obviously due to the above results
(see also \cite{TsamparlisGRG}) the conjecture/results in these papers
should be revised and the word `conformal' should be replaced with the word
`homothetic'.

It would be of interest to examine if the above close relation of the Lie
and the Noether symmetries of the second order ODEs of the form (\ref{de.1})
with the collineations of the metric is possible to be carried over to some
types of second order partial differential equations (PDEs). Although to
this question it is not possible to give a global answer, due to the
complexity of the study and the great variety of PDEs, it is still possible
to give an answer of some generality which concerns many interesting and
important cases. We will do this in the remaining sections.

\section{The case of the second-order PDE's}

\label{The case of the second order PDE's}

In the attempt to establish a general relation between the Lie symmetries of
a second order PDE and the collineations of a Riemannian space we derive the
Lie symmetry conditions for a second order PDE of the form 
\begin{equation}
A^{ij}u_{ij}-F(x^{i},u,u_{i})=0  \label{GPE.0}
\end{equation}%
and we consider the coefficients $A^{ij}(x,u)$ to be the components of a
metric in a Riemannian space. According to the standard approach \cite%
{Bluman book ODES,Olver Book,Stephani book ODES} the symmetry condition is 
\begin{equation}
X^{[2]}(H)=\lambda H  \label{GPE.10}
\end{equation}%
where $\lambda (x^{i},u,u_{i})$ is a function to be determined. $X^{[2]}$ is
the second prolongation of the Lie symmetry vector 
\begin{equation}
X=\xi ^{i}\left( x^{i},u\right) \frac{\partial }{\partial x^{i}}+\eta \left(
x^{i},u\right) \frac{\partial }{\partial u}  \label{GPE.10.1}
\end{equation}%
given by the expression: 
\begin{equation}
X^{^{[2]}}=\xi ^{i}\frac{\partial }{\partial x^{i}}+\eta \frac{\partial }{%
\partial u}+\eta _{i}^{(1)}\frac{\partial }{\partial u_{i}}+\eta
_{ij}^{\left( 2\right) }\frac{\partial }{\partial u_{ij}}  \label{GPE.10a}
\end{equation}%
where%
\begin{eqnarray*}
\eta _{i}^{(1)} &=&\frac{D\eta }{Dx^{i}}-u_{j}\frac{D\xi ^{j}}{Dx^{i}}=\eta
_{,i}+u_{i}\eta _{u}-\xi _{,i}^{j}u_{j}-u_{i}u_{j}\xi _{,u}^{j} \\
\eta _{ij}^{\left( 2\right) } &=&\frac{D\eta _{i}^{\left( 1\right) }}{Dx^{j}}%
-u_{jk}\frac{D\xi ^{k}}{Dx^{j}}=\eta _{ij}+(\eta _{ui}u_{j}+\eta
_{uj}u_{i})-\xi _{,ij}^{k}u_{k}+\eta _{uu}u_{i}u_{j}-(\xi
_{.,ui}^{k}u_{j}+\xi _{.,uj}^{k}u_{i})u_{k} \\
&&+\eta _{u}u_{ij}-(\xi _{.,i}^{k}u_{jk}+\xi _{.,j}^{a}u_{ik})-\left(
u_{ij}u_{k}+u_{i}u_{jk}+u_{ik}u_{j}\right) \xi _{.,u}^{k}-u_{i}u_{j}u_{k}\xi
_{uu}^{k}.
\end{eqnarray*}%
The introduction of the function $\lambda (x^{i},u,u_{i})$ in (\ref{GPE.10})
causes the variables $x^{i},u,u_{i}$ to be independent\footnote{%
See Ibragimov \cite{Ibragimov book 1} p. 115}.

The symmetry condition $X^{[2]}(H)=\lambda H$ when applied to (\ref{GPE.0})
gives:%
\begin{equation}
A^{ij}\eta _{ij}^{\left( 2\right) }+\left( XA^{ij}\right)
u_{ij}-X^{[1]}(F)=\lambda (A^{ij}u_{ij}-F)  \label{GPE.13}
\end{equation}%
from which follows:%
\begin{align}
0& =A^{ij}\eta _{ij}-\eta _{,i}g^{ij}F_{,u_{j}}-X(F)+\lambda F  \notag \\
& +2A^{ij}\eta _{ui}u_{j}-A^{ij}\xi _{,ij}^{a}u_{a}-u_{i}\eta
_{u}g^{ij}F_{,u_{j}}+\xi _{,i}^{k}u_{k}g^{ij}F_{,u_{j}}  \notag \\
& +A^{ij}\eta _{uu}u_{i}u_{j}-2A^{ij}\xi _{.,uj}^{k}u_{i}u_{k}+u_{i}u_{k}\xi
_{,u}^{k}g^{ij}F_{,u_{j}}  \notag \\
& +A^{ij}\eta _{u}u_{ij}-2A^{ij}\xi _{.,i}^{k}u_{jk}+(\xi
^{k}A_{,k}^{ij}+\eta A_{,u}^{ij})u_{ij}-\lambda A^{ij}u_{ij}  \notag \\
& -A^{ij}\left( u_{ij}u_{a}+u_{i}u_{ja}+u_{ia}u_{j}\right) \xi
_{.,u}^{a}-u_{i}u_{j}u_{a}A^{ij}\xi _{uu}^{a}.  \label{Po.0}
\end{align}

We note that we cannot deduce the symmetry conditions before we select a
specific form for the function $F.$ However we may determine the conditions
which are due to the second derivative of $u$ because in these terms no $F$
terms are involved. This observation significantly reduces the complexity of
the remaining symmetry condition. Following this observation we have the
condition:%
\begin{align*}
0 &
=A^{ij}\eta_{u}u_{ij}-A^{ij}(\xi_{.,i}^{k}u_{ja}+\xi_{.,j}^{k}u_{ik})+(%
\xi^{k}A_{,k}^{ij}+\eta A_{,u}^{ij})u_{ij}-\lambda A^{ij}u_{ij} \\
& -A^{ij}\left( u_{ij}u_{a}+u_{i}u_{ja}+u_{ia}u_{j}\right)
\xi_{.,u}^{a}-u_{i}u_{j}u_{a}A^{ij}\xi_{uu}^{a}
\end{align*}
from which follow the equations:%
\begin{align*}
A^{ij}\left( u_{ij}u_{k}+u_{jk}u_{i}+u_{ik}u_{j}\right) \xi_{.,u}^{k} & =0 \\
A^{ij}\eta_{u}u_{ij}-A^{ij}(\xi_{.,i}^{k}u_{jk}+\xi_{.,j}^{k}u_{ik})+(\xi
^{k}A_{,k}^{ij}+\eta A_{,u}^{ij})u_{ij}-\lambda A^{ij}u_{ij} & =0 \\
A^{ij}\xi_{uu}^{a} & =0.
\end{align*}
The first equation is written:%
\begin{equation}
A^{ij}\xi_{.,u}^{k}+A^{kj}\xi_{.,u}^{i}+A^{ik}\xi_{.,u}^{j}=0\Leftrightarrow
A^{(ij}\xi_{.,u}^{k)}=0.  \label{Po.1}
\end{equation}
The second equation gives:%
\begin{equation}
A^{ij}\eta_{u}+\eta
A_{,u}^{ij}+\xi^{k}A_{,k}^{ij}-A^{kj}\xi_{.,k}^{i}-A^{ik}\xi_{.,k}^{j}-%
\lambda A^{ij}=0.  \label{Po.2}
\end{equation}
and the last equation gives:%
\begin{equation}
A^{ij}\xi_{uu}^{k}=0.  \label{Po.2a}
\end{equation}

It is straightforward to show\footnote{%
We give a simple proof for $n=2$ in Appendix A. A detailed and more general
proof can be found in \cite{BlumanPaper}.} that condition (\ref{Po.1})
implies 
\begin{equation}
\xi _{.,u}^{k}=0
\end{equation}%
which is a well known result.

From the analysis so far we obtain the first result:

\begin{proposition}
\label{Coefficient xi second} For all second-order PDEs\ of the form $%
A^{ij}u_{ij}-F(x^{i},u,u_{i})=0,$ for which at least one of the $A^{ij}$ is $%
\neq 0$ the $\xi _{.,u}^{i}=0$ or $\xi ^{i}=\xi ^{i}(x^{j}).$ Furthermore
condition (\ref{Po.2a}) is identically satisfied.
\end{proposition}

There remains the third symmetry condition (\ref{Po.2}). We consider the
following cases:\newline
$i,j\neq0$\newline
We write(\ref{Po.2}) in an alternative form by considering $A^{ij}$ to be a
metric as follows:%
\begin{equation}
L_{\xi^{i}\partial_{i}}A^{ij}=\lambda A^{ij}-(\eta A^{ij})_{,u}
\label{GPE.32}
\end{equation}
from which follows:

\begin{proposition}
\label{Coefficient xi third} For all second-order PDEs \ of the form $%
A^{ij}u_{ij}-F(x^{i},u,u_{i})=0,$ for which $A^{ij}{}_{,u}=0$ i.e. $%
A^{ij}=A^{ij}(x^{i}),$ the vector $\xi ^{i}\partial _{i}$ is a CKV of the
metric $A^{ij}$ with conformal factor ($\lambda -\eta _{u})(x).$
\end{proposition}

Assuming\footnote{%
The index $t$ refers to the coordinate $x^{0}$ whenever it is involved.} $%
A^{tt}=A^{ti}=0$ we have\newline
- for $i=j=0$ nothing\newline
- for $i,j\neq 0$ gives (\ref{GPE.32}) and\newline
- for $i=0,j\neq 0$ becomes: 
\begin{align}
A^{tj}\eta _{u}+\eta A_{,u}^{tj}+\xi ^{k}A_{,k}^{tj}-A^{kj}\xi
_{.,k}^{t}-A^{tk}\xi _{.,k}^{j}-\lambda A^{tj}& =0\Rightarrow  \notag \\
A^{kj}\xi _{.,k}^{t}& =0  \label{GPE.30b}
\end{align}%
which leads to the following general result.

\begin{proposition}
\label{Coefficient xi fourth} For all second-order PDEs \ of the form $%
A^{ij}u_{ij}-F(x^{i},u,u_{i})=0,$ for which $A^{kj}$ is nondegenerate the $%
\xi _{.,k}^{t}=0$, that is, $\xi ^{t}=\xi ^{t}(t)$ .
\end{proposition}

Using that $\xi _{,u}^{i}=0$ when at least one of the $A_{ij}\neq 0,$ the
symmetry condition (\ref{Po.0}) is simplified as follows 
\begin{align}
0& =A^{ij}\eta _{ij}-\eta _{,i}A^{ij}F_{,u_{j}}-X(F)+\lambda F  \notag \\
& +2A^{ij}\eta _{ui}u_{j}-A^{ij}\xi _{,ij}^{a}u_{a}-u_{i}\eta
_{u}A^{ij}F_{,u_{j}}+\xi _{,i}^{k}u_{k}A^{ij}F_{,u_{j}}  \notag \\
& +A^{ij}\eta _{uu}u_{i}u_{j}+A^{ij}\eta _{u}u_{ij}-2A^{ij}\xi
_{.,i}^{k}u_{jk}  \label{GPE.30} \\
& +(\xi ^{k}A_{,k}^{ij}+\eta A_{,u}^{ij})u_{ij}-\lambda A^{ij}u_{ij}  \notag
\end{align}%
which together with the condition (\ref{GPE.32}) are the complete set of
conditions\emph{\ for all }second-order PDEs \ of the form $%
A^{ij}u_{ij}-F(x^{i},u,u_{i})=0,$ for which at least one of the $A_{ij}\neq
0 $. This class of PDEs is quite general. This fact makes the above result
very useful..

In order to continue we need to consider special forms for the function $%
F(x,u,u_{i}).$

\section{The Lie symmetry conditions for a linear function $F(x,u,u_{i})$}

\label{The Lie symmetry conditions for a linear function}

We consider the function $F(x,u,u_{i})$ to be linear in $u_{i},$ that is to
be of the form%
\begin{equation}
F(x,u,u_{i})=B^{k}(x,u)u_{k}+f(x,u)  \label{GPE.30a}
\end{equation}%
where $B^{k}(x,u),f(x,u)$ are arbitrary functions of their arguments. In
this case the PDE\ is of the form%
\begin{equation}
A^{ij}u_{ij}-B^{k}(x,u)u_{k}-f(x,u)=0.  \label{GPE.30.1}
\end{equation}

The Lie symmetries of this type of PDEs have been studied previously by
Ibragimov\cite{Ibragimov book 1}. Assuming that at least one of the $%
A_{ij}\neq 0$ the Lie symmetry conditions are (\ref{GPE.30}) and (\ref%
{GPE.32}).

Replacing $F(x,u,u_{1})$ in (\ref{GPE.30}) we find\footnote{%
We ignore the terms with $u_{ij}$ because we have already use them to obtain
condition (\ref{GPE.32}). Indeed it is easy to see that these terms give $%
A^{ij}\eta_{u}-2A^{ij}\xi_{.,i}^{k}+\xi^{k}A_{,k}^{ij}+\eta
A_{,u}^{ij}-\lambda A^{ij}=0$ which is precisely condition (\ref{GPE.32}).}:%
\begin{align}
0 & =A^{ij}\eta_{ij}-\eta_{,i}g^{ij}B_{j}-\xi^{k}f_{,k}-\eta f_{,u}+\lambda f
\notag \\
&
+2A^{ij}\eta_{ui}u_{j}-A^{ij}\xi_{,ij}^{a}u_{a}-u_{i}\eta_{u}g^{ij}B_{j}+%
\xi_{,i}^{k}u_{k}g^{ij}B_{j}+\lambda B^{k}u_{k}-\eta
B_{,u}^{k}u_{k}-\xi^{l}B_{,l}^{k}u_{k}  \notag \\
& +A^{ik}\eta_{uu}u_{i}u_{k}  \label{GPE.33} \\
& +A^{ij}\eta_{u}u_{ij}-2A^{kj}\xi_{.,k}^{i}u_{ji}+(\xi^{k}A_{,k}^{ij}+\eta
A_{,u}^{ij})u_{ij}-\lambda A^{ij}u_{ij}
\end{align}
from which follow the equations:%
\begin{align}
A^{ij}\eta_{ij}-\eta_{,i}B^{i}-\xi^{k}f_{,k}-\eta f_{,u}+\lambda f & =0
\label{GPE.34} \\
-2A^{ik}\eta_{ui}+A^{ij}\xi_{,ij}^{k}+\eta_{u}B^{k}-\xi_{,i}^{k}B^{i}+\xi
^{i}B_{,i}^{k}-\lambda B^{k}+\eta B_{,u}^{k} & =0  \label{GPE.35} \\
A^{ik}\eta_{uu} & =0.  \label{GPE.36}
\end{align}

Equation (\ref{GPE.36}) gives (because at least one $A^{ik}\neq0!):$ 
\begin{equation}
\eta=a(x^{i})u+b(x^{i}).  \label{GPE.37}
\end{equation}
Equation (\ref{GPE.35}) gives 
\begin{equation*}
-2A^{ik}a_{,i}+aB^{k}+auB_{,u}^{k}+A^{ij}\xi _{,ij}^{k}-\xi
_{,i}^{k}B^{i}+\xi ^{i}B_{,i}^{k}-\lambda B^{k}+bB_{,u}^{k}=0.
\end{equation*}%
We summarize the above results as follows.

\label{propPDE.1} The Lie symmetry conditions for the second order PDEs of
the form 
\begin{equation}
A^{ij}u_{ij}-B^{k}(x,u)u_{k}-f(x,u)=0  \label{GPE.40}
\end{equation}%
where at least one of the $A_{ij}\neq 0$ are:%
\begin{equation}
A^{ij}(a_{ij}u+b_{ij})-(a_{,i}u+b_{,i})B^{i}-\xi
^{k}f_{,k}-auf_{,u}-bf_{,u}+\lambda f=0  \label{GPE.42}
\end{equation}%
\begin{equation}
A^{ij}\xi _{,ij}^{k}-2A^{ik}a_{,i}+aB^{k}+auB_{,u}^{k}-\xi
_{,i}^{k}B^{i}+\xi ^{i}B_{,i}^{k}-\lambda B^{k}+bB_{,u}^{k}=0  \label{GPE.43}
\end{equation}%
\begin{equation}
L_{\xi ^{i}\partial _{i}}A^{ij}=(\lambda -a)A^{ij}-\eta A^{ij}{}_{,u}
\label{GPE.44}
\end{equation}%
\begin{align}
\eta & =a(x^{i})u+b(x^{i})  \label{GPE.45} \\
\xi _{,u}^{k}& =0\Leftrightarrow \xi ^{k}(x^{i}).  \label{GPE.46}
\end{align}%
We note that for all second order PDEs of the form $%
A^{ij}u_{ij}-B^{k}(x,u)u_{k}-f(x,u)=0$ for which $A^{ij}{}_{,u}=0$ i.e. $%
A^{ij}(x^{i}),$ the $\xi ^{i}(x^{j})$ is a CKV of the metric $A^{ij}.$ Also
in this case $\lambda (x^{i}).$ This result establishes the relation between
the Lie symmetries of this type of PDEs with the collineations of the metric
defined by the coefficients $A_{ij}.$

Furthermore in case the coordinates are $t,x^{i}$ (where $i=1,...,n$) $%
A^{tt}=A^{tx^{i}}=0$ and $A^{ij}$ is a nondegenerate metric we have that%
\begin{equation}
\xi _{,i}^{t}=0\Leftrightarrow \xi ^{t}(t).  \label{GPE.46a}
\end{equation}

These symmetry relations coincide with those given in \cite{Ibragimov book 1}%
.

Finally note that equation (\ref{GPE.43}) can be written:%
\begin{equation}
A^{ij}\xi _{,ij}^{k}-2A^{ik}a_{,i}+[\xi ,B]^{k}+(a-\lambda
)B^{k}+(au+b)B_{,u}^{k}=0.  \label{GPE.47}
\end{equation}%
Having derived the Lie symmetry conditions for the type of PDEs of the form $%
A^{ij}u_{ij}-B^{k}(x,u)u_{k}-f(x,u)=0$ we continue with the computation of
the Lie symmetries of some important PDEs of this form.

Before we proceed, we state two Lemmas which will be used in the discussion
of the examples.

\begin{lemma}
\label{LemmaPDE.1}

a. In flat space (in which $\Gamma _{jk}^{i}=0)$ the following identity
holds:%
\begin{equation}
L_{\xi }\Gamma _{ij}^{k}=\xi _{,ij}^{k}.
\end{equation}%
b. For a general metric $g_{ij}$ satisfying the condition $L_{\xi
^{i}\partial _{i}}g_{ij}=-(\lambda -a)g_{ij}$ the following relation holds: 
\begin{equation}
g^{jk}L_{\xi }\Gamma _{.jk}^{i}=g^{jk}\xi _{,~jk}^{i}+\Gamma _{~,l}^{i}\xi
^{l}-\xi _{~,l}^{i}\Gamma ^{l}+(a-\lambda )\Gamma _{~}^{i}.
\end{equation}%
Proof

Using the formula%
\begin{equation*}
L_{\xi }\Gamma _{.jk}^{i}=\Gamma _{.jk,l}^{i}\xi ^{l}+\xi _{,~jk}^{i}-\xi
_{~,l}^{i}\Gamma _{.jk}^{l}+\xi _{~,j}^{s}\Gamma _{.sk}^{i}+\xi
_{~,k}^{s}\Gamma _{.sj}^{i}
\end{equation*}

we have:%
\begin{align*}
g^{jk}L_{\xi }\Gamma _{.jk}^{i}& =\Gamma _{,l}^{i}\xi ^{l}+g^{jk}\xi
_{~,,jk}^{i}-g_{~~~~,l}^{jk}\xi ^{l}\Gamma _{.jk}^{i}-\xi _{~,l}^{i}\Gamma
^{l}+2g^{jk}\xi _{~,j}^{s}\Gamma _{.sk}^{i} \\
& =\Gamma _{,l}^{i}\xi ^{l}+g^{jk}\xi _{~,~jk}^{i}-\xi _{~,l}^{i}\Gamma ^{l}
\\
& -[g^{jl}\xi _{,l}^{k}+g^{kl}\xi _{,l}^{j}-(\lambda -a))g^{jk}]\Gamma
_{.jk}^{i}+2g^{jk}\xi _{~,j}^{s}\Gamma _{~sk}^{i} \\
& =\Gamma _{,l}^{i}\xi ^{l}+g^{jk}\xi _{~,~jk}^{i}-\xi _{~,l}^{i}\Gamma
^{l}+2(g^{jl}\xi _{~,j}^{k}\Gamma _{~kl}^{i}-g^{lj}\xi _{~,j}^{k}\Gamma
_{~kl}^{i})-(\lambda -a)\Gamma ^{i} \\
& =\Gamma _{,l}^{i}\xi ^{l}-\xi _{~,l}^{i}\Gamma ^{l}+g^{jk}\xi
_{~,~jk}^{i}+(a-\lambda )\Gamma ^{i}
\end{align*}
\end{lemma}

\begin{lemma}
\label{LemmaPDE.2}

Assume that the vector $\xi ^{i}$ is a CKV of the metric $g_{ij}$ with
conformal factor $-(\lambda -a)$ i.e. $L_{\xi ^{i}\partial
_{i}}g_{ij}=-(\lambda -a)g_{ij}.$ Then the following statement is true:%
\begin{equation}
g^{jk}L_{\xi }\Gamma _{.jk}^{i}=\frac{2-n}{2}(a-\lambda )^{,i}
\end{equation}%
where $n=g^{jk}g_{kj}$ the dimension of the space.

Proof

Using the identity:%
\begin{equation}
L_{\xi }\Gamma _{.jk}^{i}=\frac{1}{2}g^{ir}\left[ \nabla _{k}L_{\xi
}g_{jr}+\nabla _{j}L_{\xi }g_{kr}-\nabla _{r}L_{\xi }g_{kj}\right]
\end{equation}%
and replacing $L_{\xi }g_{ij}=(a-\lambda )g_{ij}$ we find:%
\begin{align*}
L_{\xi }\Gamma _{.jk}^{i}& =\frac{1}{2}g^{ir}\left[ (a-\lambda
)_{,k}g_{jr}+(a-\lambda )_{,j}g_{kr}-(a-\lambda )_{,r}g_{kj}\right] \\
& =\frac{1}{2}\left[ (a-\lambda )_{,k}\delta _{j}^{i}+(a-\lambda
)_{,j}\delta _{k}^{i}-g^{ir}(a-\lambda )_{,r}g_{kj}\right] .
\end{align*}%
Contracting with $g^{jk}$ we obtain the required result.
\end{lemma}

\section{Applications}

\label{Applications}

\subsection{The wave equation for an inhomogeneous medium}

In order to show how the above considerations are applied in practice we
consider the wave equation for an inhomogeneous medium in flat 2d Newtonian
space:%
\begin{equation}
c^{2}(x^{1})u_{11}-u_{22}=0.  \label{WE.0}
\end{equation}
In this case we have:%
\begin{equation*}
A_{11}=c^{-2}(x^{1}),\text{ }A_{22}=-1,\text{ }A_{12}=0;B^{i}=0;\,f=0.
\end{equation*}

The symmetry conditions (\ref{GPE.42}) - (\ref{GPE.46}) become:%
\begin{equation}
A^{ij}(a_{ij}u+b_{ij})=0  \label{WE.1}
\end{equation}%
\begin{equation}
A^{ij}\xi _{,ij}^{k}-2A^{ik}a_{i}=0  \label{WE.2}
\end{equation}%
\begin{equation}
L_{\xi ^{i}\partial _{i}}A^{ij}=(\lambda -a)A^{ij}  \label{WE.3}
\end{equation}%
\begin{align}
\eta & =a(x^{i})u+b(x^{i})  \label{WE.4} \\
& \xi ^{k}(x).  \label{WE.5}
\end{align}%
The vector $\xi ^{i}$ is a CKV of the metric $A_{ij}=$diag$%
(c^{-2}(x^{1}),-1) $ with conformal factor $-(\lambda -a).$ This is a%
\footnote{%
Because $\det A_{ij}=-c^{2}(x^{1})\neq 0$ the inverse of $A^{ij}$ exists.}
nondegenerate 2-d metric which is conformally flat, therefore if we find the
conformal factor we will have the solution $\xi ^{i}$ and the function $%
a(x^{i}).$

We take now the metric to be the $A^{ij.}$ Then according to Lemma \ref%
{LemmaPDE.1} we have (where $\Gamma _{.jk}^{i}$ are the $\Gamma ^{\prime }$
s of the metric $A_{ij}):$ 
\begin{align}
A^{jk}L_{\xi }\Gamma _{.jk}^{i}& =\Gamma _{,l}^{i}\xi ^{l}-\xi
_{.,l}^{i}\Gamma ^{l}+A^{jk}\xi _{.,jk}^{i}-(\lambda -a)\Gamma
_{.}^{i}\Rightarrow  \notag \\
A^{jk}\xi _{.,jk}^{i}& =A^{jk}L_{\xi }\Gamma _{.jk}^{i}+\xi _{.,l}^{i}\Gamma
^{l}-\Gamma _{,l}^{i}\xi ^{l}+(\lambda -a)\Gamma _{.}^{i}  \label{WE.5b}
\end{align}%
Then the Lie symmetry condition (\ref{WE.2}) $A^{ij}\xi
_{,ij}^{k}=2A^{ik}a_{i}$ is written: 
\begin{equation}
A^{ij}(L_{\xi }\Gamma _{ij}^{k}-2\delta _{j}^{k}a_{,i})=\Gamma _{,l}^{i}\xi
^{l}-\xi _{.,l}^{i}\Gamma ^{l}-(\lambda -a)\Gamma ^{i}.  \label{WE.7}
\end{equation}%
For the metric $A_{ij}=diag(-c^{2}(x^{1}),-1)$ \ we compute $\Gamma
^{i}=\Gamma _{,l}^{i}=0$ hence the last equation becomes:%
\begin{equation}
A^{ij}(L_{\xi }\Gamma _{ij}^{k}-2\delta _{j}^{k}a_{,i})=0.  \label{WE.15}
\end{equation}%
Because the metric $A_{ij}$ is nondegenerate this implies%
\begin{equation}
L_{\xi }\Gamma _{ij}^{k}=2\delta _{(j}^{k}a_{,i)}  \label{WE.16}
\end{equation}%
which means that $\xi ^{i}$ is a projective vector of the metric $A_{ij}$
with projection function $a.$ But $\xi ^{i}$ is also a CKV of the same
metric with conformal factor $a-\lambda $ therefore $\xi ^{i}$ must be a HV
of the metric $A_{ij}.$ This implies 
\begin{equation*}
a=c_{1}\text{ a constant}
\end{equation*}%
and 
\begin{equation*}
-\lambda +c_{1}=c_{2}\Rightarrow \lambda =c_{1}-c_{2}
\end{equation*}%
where $c_{2}$ is the homothetic factor. From the remaining condition (\ref%
{WE.1}) we have $A^{ij}b_{,ij}=0$ that is, the function $b$ is a solution of
the original wave equation (\ref{WE.0}). We conclude that the Lie symmetry
vector is (see also \cite{Bluman book ODES} p. 182): 
\begin{equation}
X=\xi ^{i}\partial _{i}+(c_{1}u+b)\partial _{u}  \label{WE.22}
\end{equation}%
where $\xi ^{i}$ is a HV (not necessarily proper) of the metric $A_{ij}$ and 
$b(x^{i})$ is a solution of the wave equation. The vector $\xi ^{i}\partial
_{i}$ \ is the sum $a_{i}KV^{i}+a_{HV}HV$ where $a_{i},a_{HV}$ are constants
and $KV^{i},HV$ are any of the KVs and the HV (if it exists) of the metric $%
A_{ij}.$

In the following section we consider a new example which is the heat
conduction equation with a flux in an $n-$dimensional Riemannian space.

\subsection{The heat conduction equation with a flux in a Riemannian space}

The heat equation with a flux in an $n-$dimensional Riemannian space with
metric $g_{ij}$ is:%
\begin{equation}
H\left( u\right) =q\left( t,x^{i},u\right)  \label{HEF.01}
\end{equation}%
where 
\begin{equation*}
H\left( u\right) :=g^{ij}u_{ij}-\Gamma ^{i}u_{i}-u_{t}.
\end{equation*}%
The term $q$ \ indicates that the system exchanges energy with the
environment. In this case the Lie point symmetry vector is:%
\begin{equation*}
\mathbf{X}=\xi ^{i}\left( x^{j},u\right) \partial _{i}+\eta \left(
x^{j},u\right) \partial _{u}
\end{equation*}%
where $a=t,i$. \ For this equation we have:

\begin{equation*}
A^{tt}=0,\text{ }A^{ti}=0,\text{ }A^{ij}=g^{ij},B^{i}=%
\Gamma^{i}(t,x^{i}),B^{t}=1,f(x,u)=q\left( t,x^{k},u\right).
\end{equation*}

For this PDE the symmetry conditions (\ref{GPE.42}) - (\ref{GPE.46a}) become
:%
\begin{equation}
\eta =a(t,x^{i})u+b(t,x^{i})  \label{HEF.01.1}
\end{equation}%
\begin{equation}
\ \xi ^{t}=\xi ^{t}(t)  \label{HEF.01.3}
\end{equation}%
\begin{equation}
g^{ij}(a_{ij}u+b_{ij})-(a_{,i}u+b_{,i})\Gamma ^{i}-\left(
a_{,t}u+b_{,t}\right) +\lambda q=\xi ^{t}q_{,t}+\xi ^{k}q_{,k}+\eta q_{,u}
\label{HEF.01.4}
\end{equation}%
\begin{equation}
g^{ij}\xi _{,ij}^{k}-2g^{ik}a_{,i}+a\Gamma ^{k}-\xi _{,i}^{k}\Gamma ^{i}+\xi
^{i}\Gamma _{,i}^{k}-\lambda \Gamma ^{k}=0  \label{HEF.01.6}
\end{equation}%
\begin{equation}
L_{\xi ^{i}\partial _{i}}g_{ij}=(a-\lambda )g_{ij}.  \label{HEF.01.7}
\end{equation}

The solution of the symmetry conditions is summarized in Theorem \ref{The
Lie of the heat equation with flux}. The proof of the theorem is given in
Appendix B.

\begin{theorem}
\label{The Lie of the heat equation with flux}The Lie point symmetries of
the heat equation with flux i.e. 
\begin{equation}
g^{ij}u_{ij}-\Gamma ^{i}u_{i}-u_{t}=q\left( t,x,u\right)  \label{HEF.18}
\end{equation}%
in a $n$-dimensional Riemannian space with metric $g_{ij}$ are constructed
form the homothetic algebra of the metric as follows:

a. $Y^{i}$ is a nongradient HV/KV.\newline
The Lie point symmetry is 
\begin{equation}
X=\left( 2c_{2}\psi t+c_{1}\right) \partial _{t}+c_{2}Y^{i}\partial
_{i}+\left( a\left( t\right) u+b\left( t,x\right) \right) \partial _{u}
\label{HEF.19}
\end{equation}%
where $a(t),b\left( t,x^{k}\right) ,q\left( t,x^{k},u\right) $ must satisfy
the constraint equation%
\begin{equation}
-a_{t}u+H\left( b\right) -\left( au+b\right) q_{,u}+aq-\left( 2\psi
c_{2}qt+c_{1}q\right) _{t}-c_{2}q_{,i}Y^{i}=0.  \label{HEF.20}
\end{equation}

b. $Y^{i}=S^{,i}$ is a gradient HV/KV.\newline
The Lie point symmetry is 
\begin{equation}
X=\left( 2\psi \int Tdt+c_{1}\right) \partial _{t}+TS^{,i}\partial
_{i}+\left( \left( -\frac{1}{2}T_{,t}S+F\left( t\right) \right) u+b\left(
t,x\right) \right) \partial _{u}  \label{HEF.21}
\end{equation}%
where $F(t),T(t),b\left( t,x^{k}\right) ,q\left( t,x^{k},u\right) $ must
satisfy the constraint equation%
\begin{align}
0& =\left( -\frac{1}{2}T_{,t}\psi +\frac{1}{2}T_{,tt}S-F_{,t}\right)
u+H\left( b\right) +  \notag \\
& -\left( \left( -\frac{1}{2}T_{,t}S+F\right) u+b\right) q_{,u}+\left( -%
\frac{1}{2}T_{,t}S+F\right) q-\left( 2\psi q\int Tdt+c_{1}q\right)
_{t}-Tq_{,i}S^{,i}.  \label{HEF.22}
\end{align}
\end{theorem}

We apply theorem \ref{The Lie of the heat equation with flux} for special
forms of the function $q\left( t,x,u\right) $.

\subsubsection{The homogeneous heat equation i.e. $q\left( t,x,u\right) =0$}

In this case we have the\ following results

\begin{theorem}
\label{The Lie of the heat equation}The Lie point symmetries of the
homogeneous heat conduction equation in an $n-$dimensional Riemannian space 
\begin{equation}
g^{ij}u_{ij}-\Gamma ^{i}u_{i}-u_{t}=0  \label{LHEC.01}
\end{equation}%
are constructed from the homothetic algebra of the metric $g_{ij}$ as
follows:

(a) If $Y^{i}$ is a nongradient HV/KV of the metric $g_{ij},$ the Lie point
symmetry is 
\begin{equation}
X=\left( 2\psi c_{1}t+c_{2}\right) \partial _{t}+c_{1}Y^{i}\partial
_{i}+\left( a_{0}u+b\left( t,x^{i}\right) \right) \partial _{u}
\label{LHEC.03}
\end{equation}%
where $c_{1},c_{2},,a_{0}$ are constants and $b\left( t,x^{i}\right) $ is a
solution of the homogeneous heat equation.

(c) If $Y^{i}=S^{,i}$ is a gradient HV/KV of the metric $g_{ij}$ the Lie
point symmetry is%
\begin{equation}
X=(c_{3}\psi t^{2}+c_{4}t+c_{5})\partial _{t}+(c_{3}t+c_{4})S^{i}\partial
_{i}+\left( -\frac{c_{3}}{2}S-\frac{c_{3}}{2}n\psi t+c_{5}\ \right)
u\partial _{u}+b\left( t,x^{i}\right) \partial _{u}  \label{LHEC.04+}
\end{equation}%
where $c_{3},c_{4},c_{5}$ are constants and $b\left( t,x^{i}\right) $ is a
solution of the homogeneous heat equation.
\end{theorem}

In order to compare the above result with the existing results in the
literature we consider the heat equation in a Euclidian space of dimension $%
n.$ Then in Cartesian coordinates $g_{ij}=\delta _{ij},$ $\Gamma ^{i}=0$ and
the heat equation is:%
\begin{equation}
\delta ^{ij}u_{ij}-u_{t}=0.  \label{LHEC.04}
\end{equation}%
The homothetic algebra of the space consists of the $n$ gradient KVs $%
\partial _{i}$ with generating functions $x^{i},$ the $\frac{n\left(
n-1\right) }{2}$ nongradient KVs $X_{IJ}$ which are the rotations and a
gradient HV $H^{i}$ with gradient function $H=\,R\partial _{R}.$ According
to theorem \ref{The Lie of the heat equation} the Lie symmetries of the heat
equation in the Euclidian $n$ dimensional space are (we may take $\psi =1)$%
\begin{eqnarray}
X &=&\left[ c_{3}\psi t^{2}+(c_{4}+2\psi c_{1})t+c_{5}+c_{2}\right] \partial
_{t}+\left[ c_{1}Y^{i}+(c_{3}t+c_{4})S^{i}\right] \partial _{i}+
\label{LHEC.04a} \\
&&+\left[ \left( a_{0}+\frac{c_{3}}{2}S+\frac{c_{3}}{2}n\psi t-c_{5}\right)
u+b\left( t,x^{i}\right) \right] \partial _{u}.  \notag
\end{eqnarray}%
This result agrees with the results of \cite{Stephani book ODES} p. 158.

Next we consider the de Sitter spacetime (a four dimensional space of
constant curvature and Lorentzian character) whose metric is: 
\begin{equation*}
ds^{2}=\frac{\left( -d\tau ^{2}+dx^{2}+dy^{2}+dz^{2}\right) }{\left( 1+\frac{%
K}{4}\left( -\tau ^{2}+x^{2}+y^{2}+z^{2}\right) \right) ^{2}}
\end{equation*}%
It is known that the homothetic algebra of this space consists of the ten KVs%
\begin{align*}
X_{1}& =\left( -x\tau \right) \partial _{\tau }+\left( \frac{\left( -\tau
^{2}-x^{2}+y^{2}+z^{2}\right) }{2}-\frac{2}{K}\right) \partial _{x}+\left(
-yx\right) \partial _{y}+\left( -zx\right) \partial _{x} \\
X_{2}& =\left( y\tau \right) \partial _{\tau }+\left( yx\right) \partial
_{x}+\left( \frac{\left( -x^{2}-z^{2}+y^{2}+\tau ^{2}\right) }{2}+\frac{2}{K}%
\right) \partial _{y}+\left( yz\right) \partial _{x} \\
X_{3}& =\left( z\tau \right) \partial _{\tau }+\left( zx\right) \partial
_{x}+\left( zy\right) \partial _{y}+\left( \frac{\left(
-x^{2}-y^{2}+z^{2}+\tau ^{2}\right) }{2}+\frac{2}{K}\right) \partial _{x} \\
X_{4}& =\left( \frac{\left( x^{2}+y^{2}+z^{2}+\tau ^{2}\right) }{2}-\frac{2}{%
K}\right) \partial _{\tau }+\left( \tau x\right) \partial _{x}+\left( \tau
y\right) \partial _{y}+\left( \tau z\right) \partial _{x} \\
X_{5}& =x\partial _{\tau }+\tau \partial _{x}~,~X_{6}=y\partial _{\tau
}+\tau \partial _{y}~,~X_{7}=z\partial _{\tau }+\tau \partial
_{z}~,~X_{8}=y\partial _{x}-x\partial _{y} \\
X_{9}& =z\partial _{x}-x\partial _{z}~,~X_{10}=z\partial _{y}-y\partial _{z}
\end{align*}%
all of which are nongradient. According to Theorem \ref{The Lie of the heat
equation} the Lie symmetries of the heat equation in de Sitter space are%
\begin{equation*}
\partial _{t}+\sum\limits_{A=1}^{10}c_{A}X_{A}+(a_{0}u+b\left( x,u\right)
)\partial _{u}.~
\end{equation*}

From Theorem \ref{The Lie of the heat equation} we have the following
additional results.

\begin{corollary}
The one dimensional homogenous heat equation admits a maximum number of
seven Lie symmetries (modulo a solution of the heat equation).
\end{corollary}

\begin{proof}
The homothetic group of a 1-dimensional metric $ds^{2}=g^{2}\left( x\right)
dx^{2}$ consists of one gradient KV (the $\frac{1}{g\left( x\right) }%
\partial_{x})$ and one gradient HV $(\frac{1}{g\left( x\right) }\int g\left(
x\right) dx~\partial_{x})$. According to theorem \ref{The Lie of the heat
equation} from the KV we have two Lie symmetries and from the gradient HV\
another two Lie symmetries. To these we have to add the two Lie symmetries $%
X=a_{0}u\partial_{u}+b\left( t,x^{i}\right) \partial _{u}$ and the trivial
Lie symmetry $\partial_{t}$ where $b\left( t,x^{i}\right) $ is a solution of
the heat equation.
\end{proof}

\begin{corollary}
The homogeneous heat equation in a space of constant curvature of dimension $%
n$ has at most \ $\left( n+3\right) +\frac{1}{2}n\left( n-1\right) $ (modulo
a solution of the heat equation).
\end{corollary}

\begin{proof}
A space of constant curvature of dimension $n$ admits $n+\frac{1}{2}n\left(
n-1\right) $ nongradient KVs To these we have to add the Lie symmetries \ $%
X=c\partial _{t}+a_{0}u\partial _{u}+b\left( t,x^{i}\right) \partial _{u}.$
\end{proof}

\begin{corollary}
The heat conduction equation in a space of dimension $n$ admits at most $%
\frac{1}{2}n\left( n+3\right) +5$ Lie symmetries (modulo a solution of the
heat equation) and if this is the case the space is flat.
\end{corollary}

\begin{proof}
The space with the maximum homothetic algebra is the flat space which admits 
$n$ gradient KVs, $\frac{1}{2}n\left( n-1\right) $ nongradient KVs and one
gradient HV. Therefore from Case 1, of the theorem we have $\left(
n+1\right) +\frac{1}{2}n\left( n-1\right) $ Lie symmetries. From Case 2. we
have another $\left( n+1\right) $ Lie symmetries and to these we have to add
the Lie symmetries $X=c_{1}\partial _{t}+a_{0}u\partial _{u}+b\left(
t,x^{i}\right) \partial _{u}$ where $b\left( t,x^{i}\right) $ is a solution
of the heat equation. The set of all these symmetries is $1+2n+\frac{1}{2}%
n\left( n-1\right) +2+1+1=$ $\frac{1}{2}n\left( n+3\right) +5$ \cite%
{Ibragimov book 1}.
\end{proof}

\bigskip

\subsubsection{Case $q\left( t,x,u\right) =q\left( u\right) $}

In this case we have the following result:

\begin{theorem}
The Lie symmetries of the heat equation with conduction $q(u)$ in a $n$
dimensional Riemannian space 
\begin{equation}
g^{ij}u_{ij}-\Gamma^{i}u_{i}-u_{t}=q\left( u\right)  \label{HEF.23}
\end{equation}
are constructed form the homothetic algebra of the metric as follows.

a. $Y^{i}$ is a HV/KV\newline
The Lie point symmetry is 
\begin{equation}
X=\left( 2c\psi t+c_{1}\right) \partial _{t}+cY^{i}\partial _{i}+\left(
a\left( t\right) u+b\left( t,x\right) \right) \partial _{u}  \label{HEF.24}
\end{equation}%
where the functions $a\left( t\right) ,$ $b\left( t,x\right) $ and $q\left(
u\right) $ satisfy the condition 
\begin{equation}
-a_{t}u+H\left( b\right) -\left( au+b\right) q_{,u}+\left( a-2\psi c\right)
q=0.  \label{HEF.25}
\end{equation}%
b. $Y^{i}=S^{,i}$ is a gradient HV/KV\newline
The Lie point symmetry is 
\begin{equation}
X=\left( 2\psi \int Tdt+c_{1}\right) \partial _{t}+TS^{,i}\partial
_{i}+\left( \left( -\frac{1}{2}T_{,t}S+F\left( t\right) \right) u+b\left(
t,x\right) \right) \partial _{u}  \label{HEF.26}
\end{equation}%
where $b\left( t,x\right) $ is a solution of the homogeneous heat equation,
the functions $T(t),$ $F\left( t\right) ~$and the flux $q\left( u\right) ~$%
satisfie the equation:%
\begin{equation}
\left( -\frac{1}{2}T_{,t}\psi +\frac{1}{2}T_{,tt}S-F_{,t}\right) u+H\left(
b\right) -\left( \left( -\frac{1}{2}T_{,t}S+F\right) u+b\right)
q_{,u}+\left( -\frac{1}{2}T_{,t}S+F\right) q-2\psi qT=0  \label{HEF.27}
\end{equation}
\end{theorem}

\bigskip

For various cases of $q\left( u\right) $ we obtain the results of the
following table.

\begin{center}
\begin{tabular}{|l|l|}
\hline
$q\left( u\right) $ & Lie Symmetry vector \\ \hline
$q_{0}u$ & $\left( \psi T_{0}t^{2}+2c\psi t+c_{1}\right) \partial
_{t}+\left( cY^{i}+T_{0}tS^{,i}\right) \partial _{i}+$ \\ 
& $~+\left( \left[ -2\psi cq_{0}t+a_{0}+T_{0}\left( -\frac{1}{2}S-\psi
q_{0}t^{2}-\frac{1}{2}t\right) \right] u+b\left( t,x\right) \right) \partial
_{u}\;\text{where $H(b)-bq_{0}=0$}$ \\ 
$q_{0}u^{n}$ & $\left( 2c\psi t+c_{1}\right) \partial _{t}+cY^{i}\partial
_{i}+\left( \frac{2\psi c}{1-n}u\right) \partial _{u}$ \\ 
$u\ln u$ & $c_{1}\partial _{t}+\left( Y^{i}+T_{0}e^{-t}K^{^{,}i}\right)
\partial _{i}+\left( a_{0}e^{-t}u\right) \partial _{u}~,~K^{^{,}i}$ \\ 
$e^{u}$ & $\left( 2c\psi t+c_{1}\right) \partial _{t}+cY^{i}\partial
_{i}+\left( -2\psi c\right) \partial _{u}$ \\ \hline
\end{tabular}
\end{center}

where $Y^{i}$ is a HV/KV , $S^{,i}$ is a gradient HV/KV and $K^{,i}$ is a
gradient KV.

\section{Conclusion}

\label{Conclusions}

The main result of this work is Proposition \ref{Coefficient xi third} which
states that the Lie symmetries of the PDEs of the form \ (\ref{GPE.10}) are
obtained from the conformal Killing vectors of the metric defined by the
coefficients $A_{ij},$ provided $A_{ij,u}=0.$ This result is quite general
and covers many well known and important PDEs of Physics. The geometrization
of the Lie symmetries and their association with the collineations of the
metric, dissociates their determination from the dimension of the space,
because the collineations of the metric depend (in general) on the type of
the metric and not on the dimensions of the space where the metric resides.
Furthermore, this association provides a wealth of results of Differential
Geometry on collineations which is possible to be used in the determination
of the Lie symmetries.

We have applied the above theoretical results to two cases. The first case
concerns the two dimensional wave equation in an inhomogeneous medium and
shows the application of the general results in practice. The second example
concerns the determination of the Lie symmetry vectors of the heat
conduction equation with a flux in a Riemannian space, a problem which has
not been considered before in the literature. We proved that the Lie
symmetry algebra of this PDE is generated from the homothetic algebra of the
metric. We specialized the equation to the homogeneous heat conduction
equation and regained the existing results for the Newtonian case. Finally
it can be shown that the Lie symmetries of the Poisson equation in a
Riemannian space which have been computed in \cite{Boskov} are obtained from
the present formalism in a straightforward manner.

\section*{Appendix A}

We prove the statement for $n=2$. The generalization to any $n$ is
straightforward. For a general proof see \cite{BlumanPaper}. We consider $%
A^{ij}$ as a matrix and assume that the inverse of this matrix exists. We
denote the inverse matrix with $B_{ij}$ and we get form (\ref{Po.1}):%
\begin{align}
B_{ij}A^{ij}\xi _{.,u}^{k}+B_{ij}A^{kj}\xi _{.,u}^{i}+B_{ij}A^{ik}\xi
_{.,u}^{j}& =0  \notag \\
2\xi _{.,u}^{k}+\delta _{i}^{k}\xi _{.,u}^{i}+\delta _{j}^{k}\xi _{.,u}^{j}&
=0\Rightarrow  \notag \\
\xi _{.,u}^{k}& =0.  \label{Po.3}
\end{align}

Now assume that the matrix $A^{ij}$ does not have an inverse. Then we
consider $n=2$ and write: 
\begin{equation*}
\lbrack A^{ij}]=\left[ 
\begin{tabular}{ll}
$A_{11}$ & $A_{12}$ \\ 
$A_{12}$ & $A_{22}$%
\end{tabular}
\ \ \ \ \right] \Rightarrow\det A^{ij}=A^{11}A^{22}-(A^{12})^{2}=0
\end{equation*}
where at least one of the $A^{ij}\neq0.$ Assume $A^{11}\neq0.$ Then equation
(\ref{Po.1}) for $i=j=k=1$ gives%
\begin{equation*}
3A^{11}\xi_{.,u}^{1}=0\Rightarrow\xi_{.,u}^{1}=0.
\end{equation*}
The same equation for $i=j=k=2$ gives%
\begin{equation*}
3A^{22}\xi_{.,u}^{2}=0
\end{equation*}
therefore either $\xi_{.,u}^{2}=0$ or $A^{22}=0.$ If $A^{22}=0$ then from
the condition $\det A^{ij}=0$ we have $A^{12}=0$ hence $A_{ij}=0$ which we
do not assume. Therefore $\xi_{.,u}^{2}=0.$

We consider now equations $i=j\neq k$ and find:%
\begin{equation*}
A^{ii}\xi_{.,u}^{k}+A^{ki}\xi_{.,u}^{i}+A^{ik}\xi_{.,u}^{i}=0.
\end{equation*}
Because $i\neq k$ this gives $A^{ii}\xi _{.,u}^{k}=0$ and because we have
assumed $A^{11}\neq 0$ it follows $\xi _{.,u}^{2}.$ Therefore again we find $%
\xi _{.,u}^{k}=0.$

\section{Appendix B}

Condition (\ref{HEF.01.7}) means that $\xi ^{i}$ is a CKV of the metric $%
g_{ij}$ with conformal factor $a(t,x^{k})-\lambda (t,x^{k}).$

Condition (\ref{HEF.01.6}) implies $\xi ^{k}=T\left( t\right) Y^{k}\left(
x^{j}\right) $ where $Y^{i}$ is a HV with conformal factor $\psi ,$ that is,
we have: 
\begin{equation*}
L_{Y^{i}}g_{ij}=2\psi g_{ij}\text{ ~,~}\psi \text{=constant.}
\end{equation*}%
and%
\begin{equation*}
\xi _{,t}^{t}=a-\lambda
\end{equation*}%
from which follow 
\begin{equation}
\xi ^{t}\left( t\right) =2\psi \int Tdt.  \label{HEF.01.08}
\end{equation}%
\begin{equation}
-2g^{ik}a_{,i}+T_{,t}Y^{k}=0.  \label{HEF.01.09}
\end{equation}%
Condition (\ref{HEF.01.4}) becomes:%
\begin{equation*}
H(a)u+H(b)+(a-\xi _{,t}^{t})q=\xi ^{t}q_{,t}+T(t)Y^{k}q_{,k}+\eta
q_{,u}\Rightarrow
\end{equation*}%
\begin{equation*}
H(a)u+H(b)-\left( au+b\right) q_{,u}+aq-(\xi ^{t}q)_{,t}-T(t)Y^{k}q_{,k}=0
\end{equation*}%
\begin{equation}
H\left( a\right) u+H\left( b\right) -\left( au+b\right) q_{,u}+aq-\left(
2\psi q\int Tdt\right) _{t}-Tq_{,i}Y^{i}=0.  \label{HEF.12}
\end{equation}

We consider the following cases:

Case 1 $Y^{k}$ is a HV/KV \newline
From (\ref{HEF.01.09}) we have that $T_{,t}=0\rightarrow T\left( t\right)
=c_{2}$ and $a_{,i}=0\rightarrow a\left( t,x^{k}\right) =a\left( t\right) .$
Then (\ref{HEF.12}) becomes%
\begin{equation}
-a_{t}u+H\left( b\right) -\left( au+b\right) q_{,u}+aq-\left( 2\psi
c_{2}qt+c_{1}q\right) _{t}-c_{2}q_{,i}Y^{i}=0  \label{HEF.12.1}
\end{equation}

Case 2 $Y^{k}$ is a gradient HV/KV, that is $Y^{k}=S^{,k}$ \newline
From (\ref{HEF.01.09}) we have%
\begin{equation}
a\left( t,x^{k}\right) =-\frac{1}{2}T_{,t}S+F\left( t\right) .
\label{HEF.13}
\end{equation}%
Replacing in (\ref{HEF.12}) we find the constraint equation:%
\begin{align}
0& =\left( -\frac{1}{2}T_{,t}\psi +\frac{1}{2}T_{,tt}S-F_{,t}\right)
u+H\left( b\right) +  \notag \\
& -\left( \left( -\frac{1}{2}T_{,t}S+F\right) u+b\right) q_{,u}+\left( -%
\frac{1}{2}T_{,t}S+F\right) q-\left( 2\psi q\int Tdt+c_{1}q\right)
_{t}-Tq_{,i}S^{,i}.  \label{HEF.14}
\end{align}%


\newpage

\begin{thebibliography}{99}
\bibitem{Bluman book ODES} Bluman W G and Kumei S 1989 \emph{Symmetries and
Differential Equations} (Springer Verlag New York)

\bibitem{Olver Book} Olver P J 1986 \emph{Application of Lie groups to
differential equations }(Springer Graduate texts in Mathematics, New York:
Springer)

\bibitem{Stephani book ODES} Stephani H 1989 \emph{Differential Equations:
Their Solutions using Symmetry }(Cambridge University Press)

\bibitem{Prince Crampin (1984) 1} Prince G E and Crampin M 1984 Gen.
Relativ. Gravit. \textbf{16} 921

\bibitem{Aminova 2006} Aminova A V and Aminov N A 2006 Sbornic Mathematics 
\textbf{197} 951

\bibitem{Aminova2010} Aminova A V and Aminov N A 2010 Sbornic Mathematics 
\textbf{201} 631

\bibitem{FerozeMahomedQadir} Feroze T, Mahomed F M and Qadir A 2006
Nonlinear Dynamics \textbf{45} 65

\bibitem{TsamparlisGRG} Tsamparlis M and Paliathanasis A 2010 Gen. Relativ.
Gravit. \textbf{42} 2957

\bibitem{TsamparlisGRG2} Tsamparlis M and Paliathanasis A 2011 Gen. Relativ.
Gravit. \textbf{43} 1861

\bibitem{Barnes} Barnes A 1993 Class. Quantum Grav. \textbf{10} 1139

\bibitem{Ibragimov book 1} Ibragimov Nail H 1985\emph{\ `}Transformation
Groups applied to Mathematical Physics' D. Reidel Publishing Co, Dordrecht, 1

\bibitem{Boskov} Bozhkov Y Freite I L\emph{\ }2010 J Differential Equations 
\textbf{249} 872

\bibitem{Katzin} Katzin G H, Levine J and Davis R W 1969 J. Math. Phys. 
\textbf{10} 617

\bibitem{HallR} Hall G S and Roy I M 1997 Gen. Relativ. Gravit. \textbf{29}
827

\bibitem{Sarlet} Sarlet W and Cantrijin F 1981 J. Phys. A: Math. Gen. 
\textbf{14} 479

\bibitem{Kalotas} Kalotas T M and Wybourne B G 1982 J. Phys A: Math. Gen. 
\textbf{15} 2077

\bibitem{Feroze2011} Feroze T and Hussain I 2011 Journal of Geometry and
Physics \textbf{61} 658

\bibitem{Hussain2010} Hussain I 2010 Gen. Relativ. Gravit. \textbf{42} 1791

\bibitem{Feroze2010} Feroze T 2010 Modern Phys Lett. A25

\bibitem{BlumanPaper} Bluman W G 1990 J. Math. Anal. Applic. \textbf{145} 52
\end{thebibliography}
\end{document}